\documentclass{cccg26}
\usepackage{graphicx,amssymb,amsmath,enumitem}
 \usepackage{url}
\usepackage{hyperref}

\usepackage{xcolor}
\usepackage{tikz}
\usetikzlibrary{arrows.meta, positioning, calc}
\definecolor{red_point}{RGB}{231,10,50} 

\definecolor{green_point}{RGB}{0,185,00} 

\definecolor{blue_point}{RGB}{141,160,243}
\definecolor{purple_point}{RGB}{230,141,255}
\definecolor{X_color}{RGB}{70,70,70} 
\definecolor{red_edge}{RGB}{255,0,50}

\definecolor{green_edge}{RGB}{0,185,00}

\definecolor{blue_edge}{RGB}{141,160,243}
\definecolor{purple_edge}{RGB}{230,141,255}

\definecolor{dummycolor}{RGB}{180, 120, 220}
\definecolor{dummyedgecolor}{RGB}{150, 80, 200}
\definecolor{dummybiedgecolor}{RGB}{180, 180, 180}
\definecolor{stcolor}{RGB}{240, 225, 60}
\definecolor{stcolor}{RGB}{137, 196, 244}
\definecolor{stcolor}{RGB}{235, 140, 130}
\definecolor{stcolor}{RGB}{255, 235, 100}

\definecolor{realdirected}{RGB}{255, 0, 0}

\def\nodescale{0.5}

\tikzset{
  mynode/.style={
    circle,
    draw=black, 
    ultra thin,
    fill=black!10, 
    scale = \nodescale,
    behind path
  }
}
\usepackage{float}
\title{NP-Hardness of Non-Crossing Hamiltonian Path and Cycle\\ in Non-Planar Graphs}

\author{Randal Tuggle\thanks{Dept Comp Sci, UNC Chapel Hill, \texttt{rtuggle@cs.unc.edu}}
	\and
	Jack Snoeyink\thanks{Depts Comp Sci \& SDSS, UNC CH, \texttt{snoeyink@cs.unc.edu}}}

\index{Tuggle, Randal}
\index{Snoeyink, Jack}

\pgfdeclarelayer{background}
\pgfsetlayers{background,main}

\begin{document}
\thispagestyle{empty}
\maketitle

\begin{abstract}
We seek to disentangle the hardness of finding a Hamiltonian path or cycle from the hardness of finding a non-crossing path or cycle by giving a direct reduction from 3-SAT to the non-crossing Hamiltonian path and cycle problems on non-planar graphs. Prior hardness proofs proceed by reduction to planar graphs, where every path is automatically non-crossing; this conflates the two sources of difficulty and leaves unclear why forbidding crossings on the path alone makes the problem hard. Our reduction places the difficulty squarely in the non-crossing constraint, avoids planar gadget constructions, and yields a more transparent proof that may be easier to extend to related problems.
\end{abstract}

\kern -.5em
\section{Introduction}

\kern -.5em
The Hamiltonian path problem is NP-complete for general graphs, although many restricted variants have been shown to be computationally tractable. Introducing geometric constraints, such as requiring the path to be non-crossing, can significantly alter the complexity of the problem; a substantial body of work studies non-crossing paths and related structures under geometric constraints~\cite{ABELLANAS1999141, bandy, eppstein2023non, JAPAN, longAlt, soukup2024bicoloredpointsetsadmitting, tuggle2025enumeration}. One notable example is the Hamiltonian path problem in  a complete bipartite graph $K_{m,n}$: Allowing crossings, a path exists iff $|m-n| \leq 1$, but the non-crossing Hamiltonian path problem in an embedded $K_{m,n}$ remains open, perhaps because existing techniques to establish hardness do not rely solely on the path or cycle's geometric constraints.

Several proofs establish that finding a Hamiltonian path or cycle in a planar graph is NP-hard~\cite{garey1979computers, garey1974simplified, garey1976planar, plesnick1979}. Since all paths and cycles in a planar embedding are inherently non-crossing, the NP-hardness of finding a non-crossing Hamiltonian path or cycle in an arbitrary graph follows as a corollary. However, in each case the hardness is obscured by the gadgets required to enforce planarity of the underlying graph. The complexity is hidden in crossover and planarity-forcing constructions rather than arising transparently from the non-crossing constraint on the path/cycle itself, leaving few techniques for proving hardness in non-planar graphs.

To help gain a deeper understanding of how non-crossing constraints on the path/cycle alone affect the complexity of the Hamiltonian path/cycle problem, we provide a straightforward NP-hardness proof of the non-crossing Hamiltonian path problem in embedded directed graphs, without relying on the planarity of the underlying graph.

\kern -.5em
\section{NP-Completeness of NCHP}\label{sec:path_reduction}

\kern -.5em

In this section, we prove \textsc{3-SAT} $\leq_p$ \textsc{NCHP}, where \textsc{NCHP} denotes the problem of finding a non-crossing Hamiltonian path in a graph embedding. Our reduction avoids the need for a crossover gadget, which for many problems can be difficult to construct.

 Our reduction is inspired by the standard mapping reduction from \textsc{3-SAT} to the Hamiltonian path problem. Our approach modifies the placement of vertices and includes an alternative clause gadget to accommodate the non-crossing constraint. We define the reduction as follows: given a  \textsc{3-SAT} instance $\phi$ with $n$ variables $x_1, x_2, \ldots, x_n$ and $m$ clauses $c_1, c_2, \ldots, c_m$, the mapping $f_{3SAT \to NCHP}$ constructs a corresponding embedded graph~$G_\phi$.  Assume $\phi$ contains an even number of clauses by duplicating a clause if necessary.  This makes the construction more symmetric. 
 
\kern -.5em
\subsection{Graph Construction}
\kern -.25em
We represent each variable $x_i$ as an undirected ring of $3(m+1)$ vertices $v_{(i,0)}, v_{(i,1)}, \ldots, v_{(i,3m+2)}$. That is, for every $0 \leq j < 3m+2$, there are edges from $v_{(i,j)}$ to $v_{(i,j+1)}$ and from $v_{(i,j+1)}$ to $v_{(i,j)}$. The rings are arranged concentrically, with the outermost ring for $x_1$ and innermost for $x_n$.  Clockwise traversal of a ring (beginning at $v_{(i,0)}$) signifies setting $x_i$ to true;  counterclockwise traversal (beginning at $v_{(i,3m+2)}$) signifies setting $x_i$ to false. 

As Figure~\ref{fig:ring_literal} illustrates, directed edges from a start vertex $s$ connect to both ends of the $x_1$ ring and, for all $1 \leq i < n$, each end of ring $x_i$ connects to both ends of ring $x_{i+1}$.  Both ends of the $x_n$ ring  connect to a terminal vertex $t$, ensuring that every Hamiltonian path begins at $s$ and ends at $t$. All figures will use this convention:

\newtheorem{convention}{Convention}
\begin{convention}\label{conv:arrows} In the graph $G_\phi$, solid black lines are bidirectional edges;  red arrows are unidirectional edges. \end{convention}

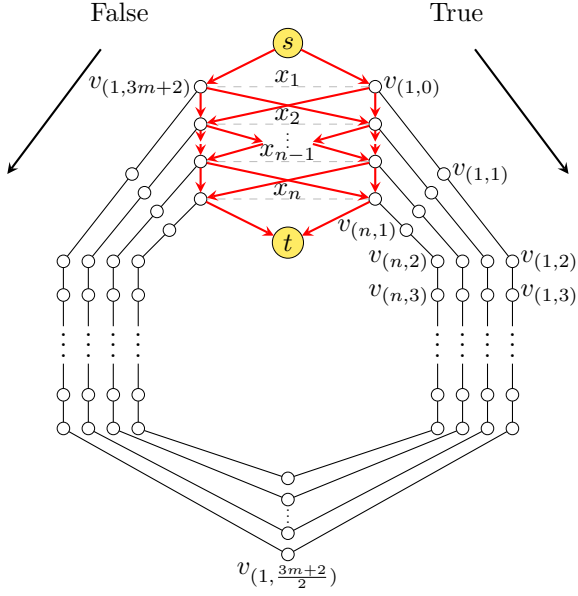
\begin{figure}[H]
\hspace{.4cm}\begin{minipage}[c]{0.8\textwidth}
\begin{tikzpicture}[scale=1.65]

\def\armreach{2}
\def\ringstep{.2}
\def\arrowendx{2.25}   
\def\arrowendy{-1}  
\def\bottomsepA{.17}   
\def\bottomsepB{.17}   
\def\bottomdotgap{.27} 
\def\endpointx{.7}
\def\endpointsep{.3}
\def\armsepA{.27}
\def\armsepB{.27}
\def\dotgap{.8}
\def\startgap{0.35}
\def\endgap{0.35}
\def\arclabelx{1.35}
\def\arcstartxparam{1.5}
\def\arcshifty{-0.5}
\def\endpointshift{0.5}   

\pgfmathsetmacro{\armtopy}{-4*\endpointsep - 1}
\pgfmathsetmacro{\starty}{\startgap - \endpointsep - \endpointshift}
\pgfmathsetmacro{\endy}{-4*\endpointsep - \endgap - \endpointshift}
\pgfmathsetmacro{\dotsy}{(-2.3)*\endpointsep - \endpointshift}

\node [mynode,fill=stcolor] (start) at (0,{\starty}) {\scalebox{2}{$s$}};
\node [mynode,fill=stcolor] (end) at (0,{\endy}) {\scalebox{2}{$t$}};
\node [mynode,fill=white,draw=none] at (0,{\dotsy}) {\scalebox{1.25}{$\vdots$}};

\pgfmathsetmacro{\arclabely}{\arcshifty + 0.3}
\draw[thick,->,>=stealth] (\arcstartxparam,{\arcshifty}) -- (\arrowendx,{\arcshifty + \arrowendy});
\draw[thick,->,>=stealth] (-\arcstartxparam,{\arcshifty}) -- (-\arrowendx,{\arcshifty + \arrowendy});
\node at (\arclabelx,{\arclabely}) { True};
\node at (-\arclabelx,{\arclabely}) { False};

\foreach \s in {1,2,3,4} {
    \pgfmathsetmacro{\xval}{\armreach - \ringstep*\s}
    \pgfmathsetmacro{\yA}{\armtopy}
    \pgfmathsetmacro{\yB}{\armtopy - \armsepA}
    \node [mynode,fill=white] (v\s1) at ({\xval},{\yA}) {$ $};
    \node [mynode,fill=white] (v\s2) at ({\xval},{\yB}) {$ $};
}

\foreach \s in {1,2,3,4} {
    \pgfmathsetmacro{\xval}{\armreach - \ringstep*\s}
    \pgfmathsetmacro{\yC}{\armtopy - \armsepA - \dotgap}
    \pgfmathsetmacro{\yD}{\armtopy - \armsepA - \dotgap - \armsepB}
    \node [mynode,fill=white] (v\s3) at ({\xval},{\yC}) {$ $};
    \node [mynode,fill=white] (v\s4) at ({\xval},{\yD}) {$ $};
}

\foreach \s in {1,2,3,4} {
    \pgfmathsetmacro{\xval}{-\armreach + \ringstep*\s}
    \pgfmathsetmacro{\yA}{\armtopy}
    \pgfmathsetmacro{\yB}{\armtopy - \armsepA}
    \node [mynode,fill=white] (v\s9) at ({\xval},{\yA}) {$ $};
    \node [mynode,fill=white] (v\s8) at ({\xval},{\yB}) {$ $};
}

\foreach \s in {1,2,3,4} {
    \pgfmathsetmacro{\xval}{-\armreach + \ringstep*\s}
    \pgfmathsetmacro{\yC}{\armtopy - \armsepA - \dotgap}
    \pgfmathsetmacro{\yD}{\armtopy - \armsepA - \dotgap - \armsepB}
    \node [mynode,fill=white] (v\s7) at ({\xval},{\yC}) {$ $};
    \node [mynode,fill=white] (v\s6) at ({\xval},{\yD}) {$ $};
}

\pgfmathsetmacro{\armbot}{\armtopy - \armsepA - \dotgap - \armsepB}
\def\bottomoffset{.4}   
\pgfmathsetmacro{\ybfour}{\armbot - \bottomoffset}
\pgfmathsetmacro{\ybthree}{\ybfour - \bottomsepB}
\pgfmathsetmacro{\ybtwo}{\ybthree - \bottomdotgap}
\pgfmathsetmacro{\ybone}{\ybtwo - \bottomsepA}

\node [mynode,fill=white] (v15) at (0, {\ybone})   {$ $};
\node[below=-1mm of v15] { $v_{(1,\frac{3m+2}{2})}$};
\node [mynode,fill=white] (v25) at (0, {\ybtwo})   {$ $};
\node [mynode,fill=white] (v35) at (0, {\ybthree}) {$ $};
\node [mynode,fill=white] (v45) at (0, {\ybfour})  {$ $};

\pgfmathsetmacro{\ymidbot}{(\ybtwo + \ybthree)/2}
\node at (0, {\ymidbot+.04}) {\scalebox{.6}{$\vdots$}};

\foreach \s in {1,...,4} {
    \pgfmathsetmacro{\yep}{-\s*\endpointsep - \endpointshift}
    \node [mynode,fill=white] (v\s0) at (\endpointx, {\yep}) {$ $};
    \ifnum \s=1
        \node[right=-1mm of v\s0] {  $v_{(1,0)}$};
    \fi
}

\foreach \s in {1,...,4} {
    \pgfmathsetmacro{\yep}{-\s*\endpointsep - \endpointshift}
    \node [mynode,fill=white] (v\s10) at (-\endpointx, {\yep}) {$ $};
    \ifnum \s=1
        \node[left=-1mm of v\s10] {  $v_{(1,3m+2)}$};
    \fi
}

\foreach \s in {1,2,3,4} {
    \pgfmathsetmacro{\xmid}{(\endpointx + \armreach - \ringstep*\s)/2}
    \pgfmathsetmacro{\ymid}{(-\s*\endpointsep - \endpointshift + \armtopy)/2}
    \node [mynode,fill=white] (vnewR\s) at ({\xmid},{\ymid}) {$ $};
}
\node[right=-1mm of vnewR1] {  $v_{(1,1)}$};
\node[left=-1mm of vnewR4] {  $v_{(n,1)}$};

\foreach \s in {1,2,3,4} {
    \pgfmathsetmacro{\xmid}{-(\endpointx + \armreach - \ringstep*\s)/2}
    \pgfmathsetmacro{\ymid}{(-\s*\endpointsep - \endpointshift + \armtopy)/2}
    \node [mynode,fill=white] (vnewL\s) at ({\xmid},{\ymid}) {$ $};
}

\foreach \s in {1,2,3,4} {
    \draw[-] (v\s1) -- (v\s2);
    \draw[-] (v\s8) -- (v\s9);
}

\foreach \s in {1,2,3,4} {
    \draw[-] (v\s3) -- (v\s4);
    \draw[-] (v\s4) -- (v\s5);
    \draw[-] (v\s5) -- (v\s6);
    \draw[-] (v\s6) -- (v\s7);
}

\foreach \s in {1,2,3,4} {
    \draw[-] (v\s2) -- ++(0,-.25);
    \draw[-] (v\s3) -- ++(0,+.25);
    \pgfmathsetmacro{\xval}{\armreach - \ringstep*\s}
    \pgfmathsetmacro{\ymid}{\armtopy - \armsepA - \dotgap/2}
    \node at ({\xval},{\ymid + .065}) {$\vdots$};
    \draw[-] (v\s8) -- ++(0,-.25);
    \draw[-] (v\s7) -- ++(0,+.25);
    \pgfmathsetmacro{\xlval}{-\armreach + \ringstep*\s}
    \node at ({\xlval},{\ymid + .065}) {$\vdots$};
}

\foreach \s in {1,2,3,4} {
    \draw[-] (v\s0)    -- (vnewR\s);
    \draw[-] (vnewR\s) -- (v\s1);
    \draw[-] (v\s9)    -- (vnewL\s);
    \draw[-] (vnewL\s) -- (v\s10);
}

\draw[dashed,color=lightgray] (v10)  -- (v110)
    node[midway,text=black,yshift=3pt] {  $x_{1}$};
\draw[dashed,color=lightgray] (v20)  -- (v210)
    node[midway,text=black,yshift=3pt] {  $x_{2}$};
\draw[dashed,color=lightgray] (v30)  -- (v310)
    node[midway,text=black,yshift=3pt] {  $x_{n-1}$};
\draw[dashed,color=lightgray] (v40)  -- (v410)
    node[midway,text=black,yshift=3pt] {  $x_{n}$};

\foreach \s in {1,3} {
    \pgfmathtruncatemacro{\nexts}{\s+1}
    \draw[-stealth, thick, realdirected] (v\s0)  -- (v\nexts0);
    \draw[-stealth, thick, realdirected] (v\s0)  -- (v\nexts10);
    \draw[-stealth, thick, realdirected] (v\s10) -- (v\nexts0);
    \draw[-stealth, thick, realdirected] (v\s10) -- (v\nexts10);
}

\pgfmathsetmacro{\arrowdy}{0.47*\endpointsep}
\pgfmathsetmacro{\arrowdydown}{0.47*\endpointsep}

\draw[-stealth, thick, realdirected] (v20)  -- ++(-.5,{-\arrowdydown});
\draw[-stealth, thick, realdirected] (v20)  -- ++(0,{-\arrowdydown});
\draw[-stealth, thick, realdirected] (v210) -- ++(.5,{-\arrowdydown});
\draw[-stealth, thick, realdirected] (v210) -- ++(0,{-\arrowdydown});

\draw[stealth-, thick, realdirected] (v30)  -- ++(-.5,{\arrowdy});
\draw[stealth-, thick, realdirected] (v30)  -- ++(0,{\arrowdy});
\draw[stealth-, thick, realdirected] (v310) -- ++(.5,{\arrowdy});
\draw[stealth-, thick, realdirected] (v310) -- ++(0,{\arrowdy});

\draw[-stealth, thick, realdirected] (start) -- (v10);
\draw[-stealth, thick, realdirected] (start) -- (v110);

\draw[-stealth, thick, realdirected] (v40)  -- (end);
\draw[-stealth, thick, realdirected] (v410) -- (end);

\node[right=-1mm of v11] {  $v_{(1,2)}$};
\node[right=-1mm of v12] {  $v_{(1,3)}$};

\node[left=-1mm of v41] {  $v_{(n,2)}$};
\node[left=-1mm of v42] {  $v_{(n,3)}$};

\end{tikzpicture}
\end{minipage}
\caption{Layout of the start vertex $s$, terminal vertex $t$, and concentric rings
representing variables, with arrows indicating clockwise (true) and
counterclockwise (false) traversal.}
\label{fig:ring_literal}
\end{figure}
Figure \ref{fig:clause_gadget} illustrates that, for each clause $c_j$, we make a directed chain of $n+1$ vertices $c_{(0,j)}, c_{(1,j)}, \ldots, c_{(n,j)}$. For every $0 \leq i < n$, there is a directed edge from $c_{(i,j)}$ to $c_{(i+1,j)}$. 

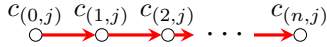
\begin{figure}[H]
    \centering
\begin{tikzpicture}[scale=1.75]
    \node [mynode,fill=white,draw=none] at (-.55,0) {\scalebox{2.5}{$\cdots$}};
    \node [mynode,fill=white] (clause0) at (-2,0) {$ $};
    \node[above] at (-2,0) {$c_{(0,j)}$};
    \node [mynode,fill=white] (clause1) at (-1.5,0) {$ $};
    \node[above] at (-1.5,0) {$c_{(1,j)}$};
    \node [mynode,fill=white] (clause2) at (-1,0) {$ $};
    \node[above] at (-1,0) {$c_{(2,j)}$};
    \node [mynode,fill=white] (clausen) at (0,0) {$ $};
    \node[above] at (0,0) {$c_{(n,j)}$};
    
    \draw[-stealth, very thick, realdirected] (clause0) -- (clause1);
    \draw[-stealth, very thick, realdirected] (clause1) -- (clause2);
    \draw[-stealth, very thick, realdirected] (clause2.east) -- ++(.15,0);
    \draw[stealth-, very thick, realdirected] (clausen) -- ++(-.35,0);
\end{tikzpicture}
    \caption{Gadget for clause $c_j$ (which may be mirrored)}
    \label{fig:clause_gadget}
\end{figure}

To connect literals to a clause $c_j$, we place the clause vertices within the innermost ring such that:
\begin{itemize}
    \item $c_{(0,j)}$ is the closest vertex in the clause gadget to the edge $\overline{v_{(n,3j-1)}v_{(n,3j)}}$,
    \item the column of vertices $c_{(0,j)}, \dots, c_{(n,j)}$ is between and parallel to the two columns of vertices $v_{(1,3j-1)}, \dots, v_{(n,3j-1)}$ and $v_{(1,3j)}, \dots, v_{(n,3j)}$, 
    \item neither edge $\overline{v_{(n,3j-1)}c_{(n,j)}}$ nor $\overline{v_{(n,3j)}c_{(n,j)}}$ intersects either of the two edges containing $t$,
    \item if $v_{(1,3j)},\ldots,v_{(n,3j)}$ are to the left of $t$, then $c_{(0,j)},…,c_{(n,j)}$ are all to the left of $t$, and 
\item if $v_{(1,3j)},\ldots,v_{(n,3j)} $ are to the right of $t$, then $c_{(0,j)},\ldots,c_{(n,j)} $ are all to the right of $t$.
\end{itemize}  Once the clause gadget is placed, we add edges connecting it to the rings. More precisely, for each $1 \leq i \leq n$, we create bi-directional edges between $v_{(i,3j-1)}$ and $c_{(i,j)}$ as well as between $v_{(i,3j)}$ and $c_{(i,j)}$. As Figure \ref{fig:enter-clause} suggests, when a literal $x_i$ is in $c_j$, we create a directed edge from $v_{(i,3j-1)}$ to $c_{(0,j)}$, and when a literal $\neg x_i$ is in $c_j$, we create a directed edge from $v_{(i,3j)}$ to $c_{(0,j)}$.  This concludes the construction of $G_\phi$. For clarity, we provide an example in Figure \ref{fig:example}.

\begin{figure}[H]
    \centering
    \pgfmathsetmacro{\xbound}{2.35}  
    
    \begin{tikzpicture}[scale=1.25]

    \node [mynode,fill=white,draw=none] at (-\xbound,1.87) {\scalebox{1}{$\vdots$}};
    \node [mynode,fill=white,draw=none] at (\xbound,1.87) {\scalebox{1}{$\vdots$}}; 
    \node [mynode,fill=white,draw=none]  at (-\xbound,1.57) {\scalebox{1}{$\vdots$}};
    \node [mynode,fill=white,draw=none]  at (\xbound,1.57) {\scalebox{1}{$\vdots$}};
    \node [mynode,fill=white,draw=none] at (-\xbound,1.27) {\scalebox{1}{$\vdots$}};
    \node [mynode,fill=white,draw=none] at (\xbound,1.27) {\scalebox{1}{$\vdots$}};
    \node [mynode,fill=white,draw=none] at (-\xbound,.97) {\scalebox{1}{$\vdots$}};
    \node [mynode,fill=white,draw=none] at (\xbound,.97) {\scalebox{1}{$\vdots$}};
    
    \node [mynode,fill=white] (v1left) at (-\xbound,2) {$ $};
    \node[left] at (-\xbound,2) {$v_{(1,3j-1)}$};

    \node [mynode,fill=white] (v1right) at (\xbound,2) {$ $};
    \node[right] at (\xbound,2) {$v_{(1,3j)}$};

    \node [mynode,fill=white] (valeft) at (-\xbound,1.7) {$ $};
    \node[left] at (-\xbound,1.7) {$v_{(\alpha,3j-1)}$};
    \node [mynode,fill=white] (varight) at (\xbound,1.7) {$ $};
    \node[right] at (\xbound,1.7) {$v_{(\alpha,3j)}$};

    \node [mynode,fill=white] (vbleft) at (-\xbound,1.4) {$ $};
    \node[left] at (-\xbound,1.4) {$v_{(\beta,3j-1)}$};
    \node [mynode,fill=white] (vbright) at (\xbound,1.4) {$ $};
    \node[right] at (\xbound,1.4) {$v_{(\beta,3j)}$};

    \node [mynode,fill=white] (vcleft) at (-\xbound,1.1) {$ $};
    \node[left] at (-\xbound,1.1) {$v_{(\gamma,3j-1)}$};
    \node [mynode,fill=white] (vcright) at (\xbound,1.1) {$ $};
    \node[right] at (\xbound,1.1) {$v_{(\gamma,3j)}$};


    
    \node [mynode,fill=white] (vnleft) at (-\xbound,.8) {$ $};
    \node[left] at (-\xbound,.8) {$v_{(n,3j-1)}$};
    \node [mynode,fill=white] (vnright) at (\xbound,.8) {$ $};
    \node[right] at (\xbound,.8) {$v_{(n,3j)}$};

    \node [mynode,fill=white,draw=none] at (0,-.5) {$\vdots$};
    \node [mynode,fill=white,draw=none] at (0,-1) {$\vdots$};
    \node [mynode,fill=white,draw=none] at (0,-1.5) {$\vdots$};
    \node [mynode,fill=white,draw=none] at (0,-2) {$\vdots$};
    
    \node [mynode,fill=white] (clause0) at (0,.25) {$ $};
    \node[below=-1mm of clause0] {\scriptsize $c_{(0,j)}$};
    
    \node [mynode,fill=white] (clause1) at (0,-.25) {$ $};
    \node[below=-1mm of clause1] {\scriptsize $c_{(1,j)}$};
    
    \node [mynode,fill=white] (clausea) at (0,-.75) {$ $};
     \node[below=-1mm of clausea] {\scriptsize $c_{(\alpha,j)}$};
    
    \node [mynode,fill=white] (clauseb) at (0,-1.25) {$ $};
     \node[below=-1mm of clauseb] {\scriptsize $c_{(\beta,j)}$};
    
    \node [mynode,fill=white] (clausec) at (0,-1.75) {$ $};
     \node[below=-1mm of clausec] {\scriptsize $c_{(\gamma,j)}$};
     
    \node [mynode,fill=white] (clausen) at (0,-2.25) {$ $};
     \node[below=-1mm of clausen] {\scriptsize $c_{(n,j)}$};


     \draw[-stealth,realdirected, very thick] (valeft) -- (clause0);
     \draw[-stealth,realdirected, very thick] (vbright) -- (clause0);
     \draw[-stealth,realdirected, very thick] (vcleft) -- (clause0);
    
    \draw[-] (v1left) -- (clause1);
    \draw[-] (v1right) -- (clause1);

    \draw[-] (valeft) -- (clausea);
    \draw[-] (varight) -- (clausea);

    \draw[-] (vbleft) -- (clauseb);
    \draw[-] (vbright) -- (clauseb);

    \draw[-] (vcleft) -- (clausec);
    \draw[-] (vcright) -- (clausec);
    
    \draw[-] (vnleft) -- (clausen);
    \draw[-] (vnright) -- (clausen);
    
\end{tikzpicture}
    \caption{We depict connecting literals to clause $c_j = (x_\alpha \vee \neg x_\beta \vee x_\gamma)$ where $\alpha < \beta < \gamma$. This structure is rotated $90^\circ$ or  $-90^\circ$ depending on which side of the 
graph the clause appears. Note that, for clarity, only the connecting edges between literals and clauses are shown; the edges depicted in Figures \ref{fig:ring_literal}-\ref{fig:clause_gadget} remain intact.}
    \label{fig:enter-clause}
\end{figure}
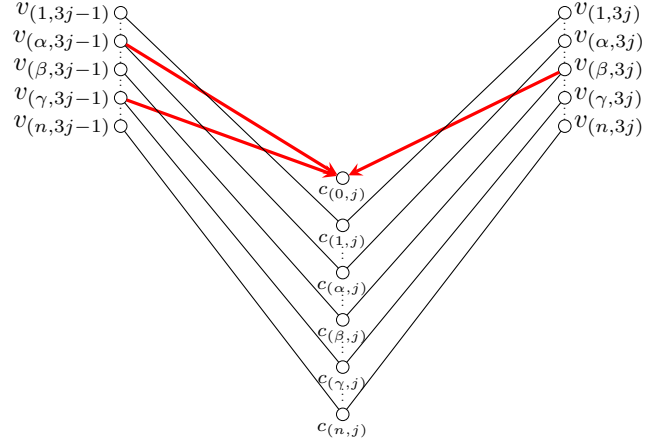

\begin{figure*}[t]
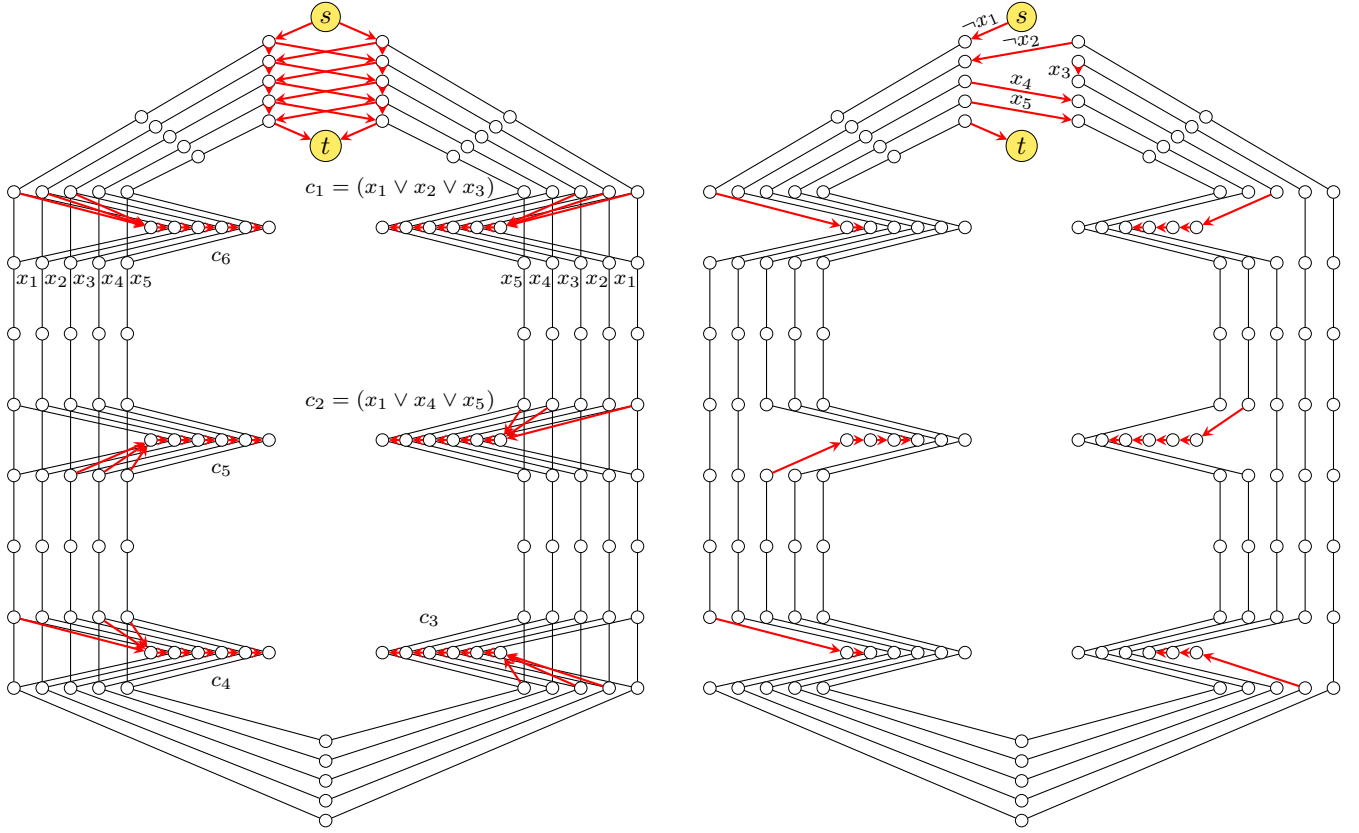

\centering
\begin{minipage}[t]{0.48\textwidth}
\input{hamiltonian_path_graph_tikz_feedback}
\end{minipage}
\hspace{1.25em}
\begin{minipage}[t]{0.48\textwidth}
\input{hamiltonian_path_graph_solution_tikz_feedback}
\end{minipage}
   \caption{Left: Highlighting the first two clauses in the graph $G_\phi$ corresponding to the formula $\phi = (x_1 \lor x_2 \lor x_3) \land\\ (x_1 \lor x_4 \lor x_5) \land (\neg x_2 \lor \neg x_3 \lor \neg x_5) \land (\neg x_1 \lor \neg x_4 \lor \neg x_5) \land (x_3 \lor x_4 \lor x_5) \land (\neg x_1 \lor \neg x_2 \lor \neg x_3)$. Right: A non-crossing Hamiltonian path corresponding to the satisfying assignment $x_1 = \text{False}, x_2 = \text{False}, x_3 = \text{True}, x_4 = \text{True}, x_5 = \text{True}$.}
    \label{fig:example}
\end{figure*}

For sufficiently large $n$ and $m$, the graph $G_\phi$ is non-planar. One can verify this by exhibiting a $K_{3,3}$ minor, taking three non-consecutive vertices from one clause column as one side of the bipartition and three non-consecutive vertices from another clause column as the other; we leave the details as an exercise.

\subsection{Hamiltonian Path Construction}

We now prove that a non-crossing Hamiltonian path can be constructed from a satisfying assignment. An example is depicted in Figure~\ref{fig:example} (right). The following theorem formalizes this result.

\begin{theorem} \label{NCHP_forward}
    If $\phi$ is satisfiable, then $G_\phi$ contains a non-crossing Hamiltonian path.
\end{theorem}

\begin{proof}
Assume $\phi$ has a satisfying assignment. To construct a non-crossing Hamiltonian path, we start at $s$ and traverse rings $i = 1,\ldots,n$ in order as follows:

If $x_i$ is true, we proceed clockwise around ring $i$ and potentially take detours to visit the clause gadgets. Let’s consider a clause $c_j$ and determine if and how we must detour into its gadget, ensuring that no crossings occur.
\filbreak
\begin{itemize}
    \item \textbf{Case 1:} $c_{(0,j)}$ has not been picked up:
    \begin{itemize}[topsep=0pt]
        \item If $x_i$ does not satisfy clause $c_j$, no detour is taken.
        
        \item If $x_i$ satisfies clause $c_j$, the path takes a detour after reaching $v_{(i,3j-1)}$. This detour collects all vertices $c_{(0,j)}, c_{(1,j)}, \ldots, c_{(i,j)}$ before rejoining ring $i$ at $v_{(i,3j)}$. Covering vertices $c_{(1,j)}, \ldots, c_{(i,j)}$ in this detour ensures that no additional detours are needed for clause $c_j$ in previously traversed rings. Additionally, because the detour is directed inward (or toward the center of the ring structure), it avoids crossing any edges from prior parts of the path.
    \end{itemize}
    
    \item \textbf{Case 2:} $c_{(0,j)}$ has already been picked up:
    \begin{itemize}[topsep=0pt]
        \item The path takes a detour after reaching $v_{(i,3j-1)}$ to avoid crossing the edge that was used to pick up $c_{(0,j)}$. This detour picks up only $c_{(i,j)}$ and rejoins at $v_{(i,3j)}$, ensuring all vertices of this clause gadget are collected without introducing crossings as we traverse the rings.
    \end{itemize}
\end{itemize}

Otherwise, $x_i$ is false and we traverse analogously in the counterclockwise direction.

After completing the traversal of ring $n$, the path concludes at $t$. To confirm that all vertices are covered: all ring vertices are visited by the traversal. Each clause gadget $c_j$, containing vertices $c_{(0,j)},\ldots,c_{(n,j)}$, is fully visited as follows. Vertices $c_{(0,j)},\ldots,c_{(i,j)}$ are picked up when the first literal (either $x_i$ or $\neg x_i$) satisfying $c_j$ is encountered. For each $i' > i$, $c_{(i',j)}$ is picked up when detouring on ring $i'$. Thus, all vertices are covered, completing a Hamiltonian path.
\end{proof}

\subsection{Non-Crossing Hamiltonian Path $\implies$ Satisfying Assignment}
We now prove  the converse in the following theorem. The main idea is that if the clause vertices were used to jump from one ring to another, then a ring vertex would be forced to become an endpoint of $\pi$, contradicting the fact that only $s$ and $t$ are endpoints.
\filbreak
\begin{theorem}\label{NCHP_backward}
    If $G_\phi$ contains a non-crossing Hamiltonian path, then $\phi$ is satisfiable.
\end{theorem}
\begin{proof}
Let $\pi$ be a non-crossing Hamiltonian path in $G_\phi$. It suffices to show that every clause detour in $\pi$ begins and ends on the same ring.

Suppose for contradiction that $\pi$ contains a subpath, or the reverse of a subpath, of the form $uCv$, where $u$ and $v$ are ring vertices on different rings and $C$ is a nonempty sequence of clause vertices. Among all such subpaths, choose one minimizing $r_1 = \min(\operatorname{ring}(u), \operatorname{ring}(v))$, with $u$ on ring $r_1$ and $v$ on ring $r_2 > r_1$.

We have $u \in \{v_{(r_1,3j-1)}, v_{(r_1,3j)}\}$ for some $j$; assume without loss of generality $u = v_{(r_1,3j-1)}$. Since any clause-gadget path from $r_1$ to $r_2 > r_1$ passes through $c_{(r_1,j)}$, we have $c_{(r_1,j)} \in C$.

We now show $v_{(r_1,3j)}$ has no valid neighbor in $\pi$. Its neighbors in $G_\phi$ are $v_{(r_1,3j-1)}$, $v_{(r_1,3j+1)}$, $c_{(r_1,j)}$, and (if $\neg x_{r_1} \in c_j$) $c_{(0,j)}$. Since $u = v_{(r_1,3j-1)}$ is not an endpoint of $\pi$, both its $\pi$-edges are used, one into $C$; if the other went to $v_{(r_1,3j)}$, then $v_{(r_1,3j-2)}$ (whose only neighbors are $v_{(r_1,3j-3)}$ and $v_{(r_1,3j-1)}$) would be forced to be an endpoint of $\pi$, a contradiction. Since $c_{(r_1,j)}$ is
interior to $uCv$, both its $\pi$-edges are used, ruling it out as well. If $v_{(r_1,3j)}$ uses $c_{(0,j)}$, then $\pi$ proceeds along $c_{(0,j)} \to c_{(1,j)}
\to \cdots$; since $c_{(r_1,j)}$ is already visited, this chain exits at some ring $r_3 < r_1$, contradicting minimality. Thus the only remaining neighbor is $v_{(r_1,3j+1)}$, forcing $v_{(r_1,3j)}$ to be an endpoint of $\pi$, but neither $s$ nor $t$ is a ring vertex, so we have a contradiction. The case $u = v_{(r_1,3j)}$ is symmetric.

Hence every clause detour stays on a single ring, each ring is traversed consistently
in one direction, and we obtain a well-defined truth assignment. Since $c_{(0,j)}$ is
reachable only via a directed literal edge from a true literal in $c_j$, every clause
is satisfied and $\phi$ is satisfiable.
\end{proof}

 So, $G_\phi$ contains a non-crossing Hamiltonian path if and only if $\phi$ is satisfiable. Since $G_\phi$ contains $3n(m+1) + m(n+1) + 2$ vertices, it can be constructed in polynomial time. Thus, \textsc{NCHP} is NP-complete.

\section{NP-Completeness of NCHC}

We now extend the reduction in Section \ref{sec:path_reduction} to the non-crossing Hamiltonian \emph{cycle} problem, which we denote \textsc{NCHC}. We show that \textsc{3-SAT} $\leq_p$ \textsc{NCHC} by augmenting the graph $G_\phi$ with additional \emph{dummy ring pairs} that provide a return path from $t$ back to $s$.

\subsection{Augmented Graph Construction}

Given a \textsc{3-SAT} instance $\phi$ with $n$ variables and $m$ clauses (with $m$ even), we construct an augmented graph $G_\phi^{\mathrm{cycle}}$ as follows. We make two modifications to the graph $G_\phi$ from Section~\ref{sec:path_reduction}: we interleave dummy clause vertices in each clause chain, and we add dummy ring pairs between consecutive real rings. Unless otherwise stated, the figures adhere to the following convention.

\begin{convention}\label{conv:arrows_cycle} In the graph $G_\phi^{\mathrm{cycle}}$, solid grey lines are bidirectional edges added to $G_\phi$; solid purple arrows are new unidirectional edges. New vertices are purple. \end{convention}

For each clause $c_j$, we introduce dummy clause vertices $e_{(i,j)}$ and $e_{(i,j)}'$ for each $1 \leq i \leq n-1$, and replace the directed edge $c_{(i,j)} \to c_{(i+1,j)}$ with the directed chain
$
c_{(i,j)} \to e_{(i,j)} \to e_{(i,j)}' \to c_{(i+1,j)}.$

For each $1 \leq i \leq n-1$, we insert a \emph{dummy pair of rings}, an outer ring $d_i$ and an inner ring $d_i'$,  geometrically between the real rings for $x_i$ and $x_{i+1}$. Each dummy ring has $3(m+1)$ vertices (e.g.,  $d_i$ contains $d_{(i,0)}, d_{(i,1)}, \ldots, d_{(i,3m+2)}$ and similarly for $d_i'$), arranged in the same layout as a real ring except at the ends. We perturb the ends $d_{(i,3m+2)}$ and $d_{(i, 3m+2)}'$ slightly to the left and perturb $d_{(i,0)}$ and $d_{(i, 0)}'$ slightly to the right so they are not collinear with the ends of the real rings. These dummy rings have the same bidirectional edges that real rings have, with these added:
\begin{itemize}
    \item \textbf{Straight edges}: for every $0 \leq k \leq 3m+2$, a bidirectional edge between $d_{(i,k)}$ and $d_{(i,k)}'$.
    \item \textbf{Diagonal edges}: for every $0 \leq k < 3m+2$, bidirectional edges between $d_{(i,k)}$ and $d_{(i,k+1)}'$, and between $d_{(i,k)}'$ and $d_{(i,k+1)}$.
    \item \textbf{Dummy clause connections}: for every $1 \leq j \leq m$, bidirectional edges between $d_{(i,3j-1)}$ and $e_{(i,j)}$, between $d_{(i,3j)}$ and $e_{(i,j)}$, between $d_{(i,3j-1)}'$ and $e_{(i,j)}'$, and between $d_{(i,3j)}'$ and $e_{(i,j)}'$. 
\end{itemize}

We add directed transition edges connecting the dummy rings to the rest of the graph:
\begin{itemize}
    \item $d_{(1,0)} \to s$ and $d_{(1,3m+2)} \to s$.
    \item $t \to d_{(n-1,0)}'$ and $t \to d_{(n-1,3m+2)}'$.
    \item For each $2 \leq i \leq n-1$, directed edges $d_{(i,0)} \to d_{(i-1,3m+2)}'$ and $d_{(i,3m+2)} \to d_{(i-1,0)}'$.

\end{itemize}

Note that the directions of the dummy ring transitions are the \emph{reverse} of the real ring transitions: real ring transitions go from outer to inner (increasing $i$), while dummy ring transitions go from inner to outer (decreasing $i$). We illustrate the augmented graph in Figure~5 (page~\pageref{fig:cycle_graph}).

\begin{figure*}[t]
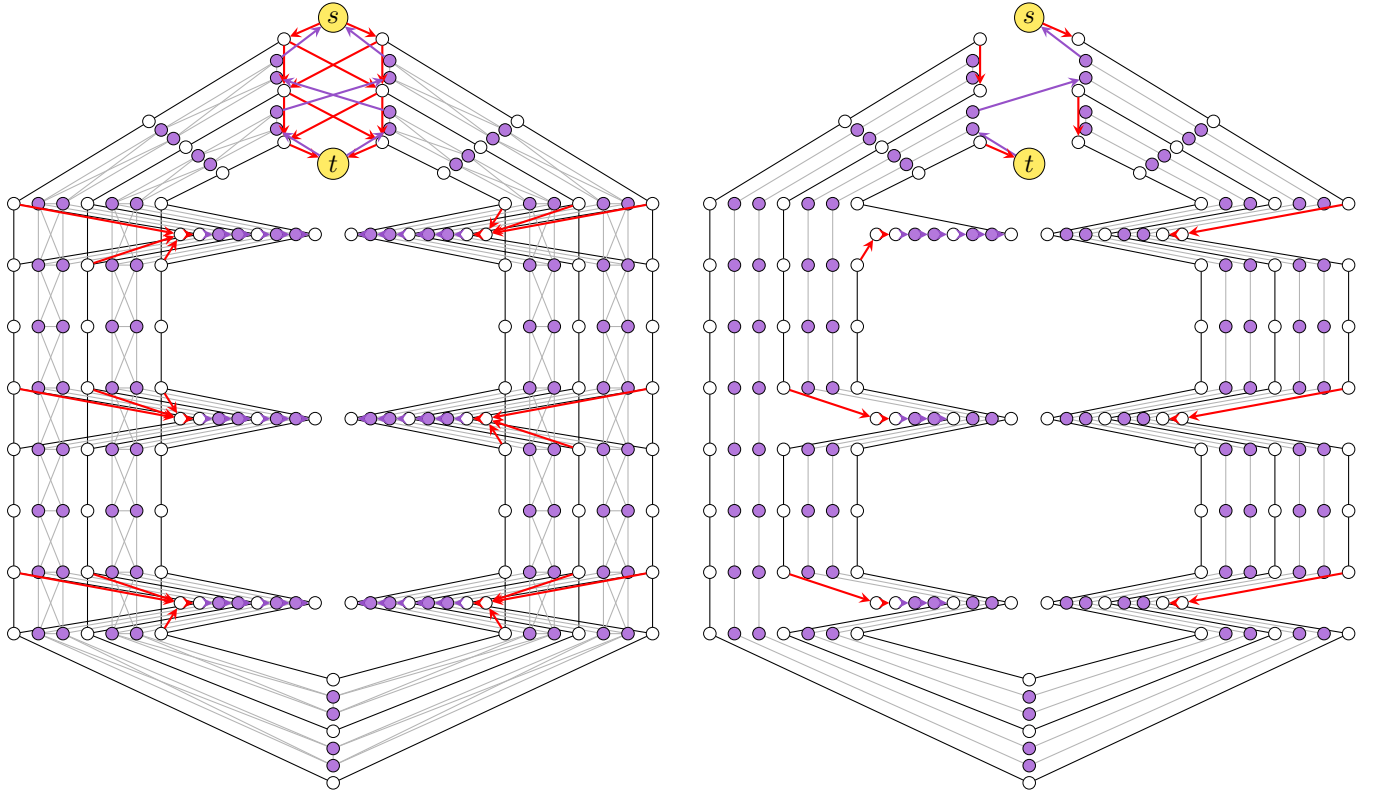

 \centering
\begin{minipage}[t]{0.48\textwidth}
\input{hamiltonian_cycle_graph_tikz}
\end{minipage}
\hspace{1.25em}
\begin{minipage}[t]{0.48\textwidth}
\input{hamiltonian_cycle_graph_solution_tikz}
\end{minipage}
   \caption{Left: The graph $G_\phi^{\mathrm{cycle}}$ corresponding to the formula $\phi = (x_1 \lor x_2 \lor x_3) \land (x_1 \lor \neg x_2 \lor \neg x_3) \land (x_1 \lor x_2 \lor \neg x_3) \land (\neg x_1 \lor \neg x_2 \lor x_3) \land (\neg x_1 \lor \neg x_2 \lor \neg x_3) \land (\neg x_1 \lor x_2 \lor x_3)$. Right: A non-crossing Hamiltonian cycle corresponding to the satisfying assignment $x_1 = \text{True}, x_2 = \text{False}, x_3 = \text{True}$.}
    \label{fig:cycle_graph}
\end{figure*}

\subsection{Hamiltonian Cycle Construction}
We prove that a non-crossing Hamiltonian cycle can be constructed from a satisfying assignment, which we also illustrate in Figure \ref{fig:cycle_graph}. The following theorem formalizes this result.

\begin{theorem}
    If $\phi$ is satisfiable, then $G_\phi^{\mathrm{cycle}}$ contains a non-crossing Hamiltonian cycle.
\end{theorem}

\begin{proof}
    
The cycle is constructed in two phases: a forward pass from $s$ to $t$ using the real rings, and a return pass from $t$ to $s$ using the dummy rings.

The forward pass proceeds identically to the Hamiltonian path construction in Theorem~\ref{NCHP_forward}, with one modification: when a path detour enters clause gadget $c_j$ for the first time at ring $x_{i_0}$ (the first satisfying ring), it follows the extended chain
\[
c_{(0,j)}, c_{(1,j)}, e_{(1,j)}, e_{(1,j)}', c_{(2,j)}, e_{(2,j)}, e_{(2,j)}', \ldots, c_{(i_0,j)},
\]
automatically collecting all dummy clause vertices $e_{(i,j)}$ and $e_{(i,j)}'$ for $1 \leq i < i_0$. After the forward pass, the only uncollected vertices are the dummy ring vertices and the dummy clause vertices $e_{(i,j)}$ and $e_{(i,j)}'$ for $i \geq i_0$.

Upon reaching $t$, the path transitions to the dummy rings via one of the directed edges $t \to d_{(n-1,\cdot)}'$, entering dummy pair $n-1$ at the endpoint opposite to where ring $x_n$ exited. The return pass then spirals outward through dummy pairs $n-1, n-2, \ldots, 1$, collecting all remaining dummy clause vertices, and finally returns to $s$ via a directed edge from $d_1$. Each dummy pair $i$ is traversed using one of four options:

\begin{enumerate}
    \item Enter at $d'_{(i,0)}$. As depicted in the upper part of Figure~\ref{fig:options}, we
    traverse clockwise along $d'_i$ to $d'_{(i,3m+2)}$,
    dipping into clause gadgets as needed. Hop up to $d_{(i,3m+2)}$. Traverse
    counterclockwise along $d_i$ back to $d_{(i,0)}$, dipping into clause gadgets as
    needed. Exit at $d_{(i,0)}$.

    \item The mirror image of Option~1. Enter at $d'_{(i,3m+2)}$, traverse
    counterclockwise along $d'_i$ to $d'_{(i,0)}$, hop up, traverse clockwise along $d_i$
    to $d_{(i,3m+2)}$, and exit at $d_{(i,3m+2)}$.

    \item Enter at $d'_{(i,0)}$. As depicted in the lower part of Figure~\ref{fig:options}, we  zig-zag between $d'_i$ and $d_i$, using straight and diagonal hops to alternate between the primed and
    unprimed rings. At each clause position $j$, a dashed diagonal hop from $d_{(i,3j-1)}$ allows a detour
    to collect $e_{(i,j)}$ and $e'_{(i,j)}$ before resuming. We collect the vertices of both rings with a single clockwise traversal and exit at $d_{(i,3m+2)}$.

    \item The mirror image of Option~3. Enter at $d'_{(i,3m+2)}$, 
     zig-zag counterclockwise, and exit at $d_{(i,0)}$.
\end{enumerate}

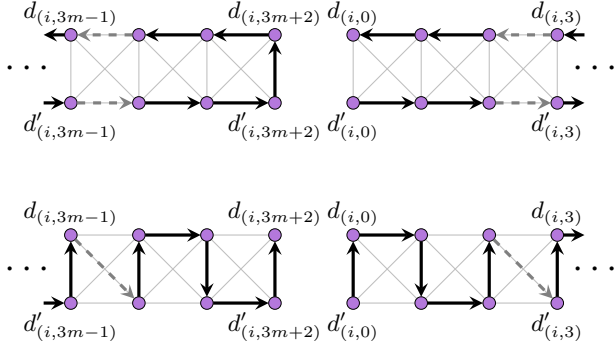
\begin{figure}[H]
    \centering
    \begin{tikzpicture}[scale=.9]
\pgfdeclarelayer{background}
\pgfsetlayers{background,main}

\def\gapsize{1.15}
\pgfmathsetmacro{\rstart}{3+\gapsize}

\node at (-0.6,.5) {\Large $\cdots$};
\node at ({\rstart+3+0.6},.5) {\Large $\cdots$};

\node[mynode,fill=dummycolor] (tL0) at (0,1) {};
\node[above,font=\small] at (0,1) {$d_{(i,3m-1)}$};
\node[mynode,fill=dummycolor] (tL1) at (1,1) {};
\node[mynode,fill=dummycolor] (tL2) at (2,1) {};
\node[mynode,fill=dummycolor] (tL3) at (3,1) {};
\node[above,font=\small] at (3,1) {$d_{(i,3m+2)}$};

\node[mynode,fill=dummycolor] (bL0) at (0,0) {};
\node[below,font=\small] at (0,0) {$d'_{(i,3m-1)}$};
\node[mynode,fill=dummycolor] (bL1) at (1,0) {};
\node[mynode,fill=dummycolor] (bL2) at (2,0) {};
\node[mynode,fill=dummycolor] (bL3) at (3,0) {};
\node[below,font=\small] at (3,0) {$d'_{(i,3m+2)}$};

\node[mynode,fill=dummycolor] (tR0) at ({\rstart},1) {};
\node[above,font=\small] at ({\rstart},1) {$d_{(i,0)}$};
\node[mynode,fill=dummycolor] (tR1) at ({\rstart+1},1) {};
\node[mynode,fill=dummycolor] (tR2) at ({\rstart+2},1) {};
\node[mynode,fill=dummycolor] (tR3) at ({\rstart+3},1) {};
\node[above,font=\small] at ({\rstart+3},1) {$d_{(i,3)}$};

\node[mynode,fill=dummycolor] (bR0) at ({\rstart},0) {};
\node[below,font=\small] at ({\rstart},0) {$d'_{(i,0)}$};
\node[mynode,fill=dummycolor] (bR1) at ({\rstart+1},0) {};
\node[mynode,fill=dummycolor] (bR2) at ({\rstart+2},0) {};
\node[mynode,fill=dummycolor] (bR3) at ({\rstart+3},0) {};
\node[below,font=\small] at ({\rstart+3},0) {$d'_{(i,3)}$};

\begin{pgfonlayer}{background}
\draw[-,lightgray] (tL0)--(tL1)--(tL2)--(tL3);
\draw[-,lightgray] (bL0)--(bL1)--(bL2)--(bL3);
\draw[-,lightgray] (tR0)--(tR1)--(tR2)--(tR3);
\draw[-,lightgray] (bR0)--(bR1)--(bR2)--(bR3);
\draw[-,lightgray] (tL0)--(bL0); \draw[-,lightgray] (tL1)--(bL1);
\draw[-,lightgray] (tL2)--(bL2); \draw[-,lightgray] (tL3)--(bL3);
\draw[-,lightgray] (tR0)--(bR0); \draw[-,lightgray] (tR1)--(bR1);
\draw[-,lightgray] (tR2)--(bR2); \draw[-,lightgray] (tR3)--(bR3);
\draw[-,lightgray] (tL0)--(bL1); \draw[-,lightgray] (tL1)--(bL2); \draw[-,lightgray] (tL2)--(bL3);
\draw[-,lightgray] (tR0)--(bR1); \draw[-,lightgray] (tR1)--(bR2); \draw[-,lightgray] (tR2)--(bR3);
\draw[-,lightgray] (bL0)--(tL1); \draw[-,lightgray] (bL1)--(tL2); \draw[-,lightgray] (bL2)--(tL3);
\draw[-,lightgray] (bR0)--(tR1); \draw[-,lightgray] (bR1)--(tR2); \draw[-,lightgray] (bR2)--(tR3);
\end{pgfonlayer}

\draw[-stealth,black,very thick] (bR0)--(bR1);
\draw[-stealth,black,very thick] (bR1)--(bR2);
\draw[-stealth,gray,very thick, dashed] (bR2)--(bR3);
\draw[-stealth,gray,very thick, dashed] (bL0)--(bL1);
\draw[-stealth,black,very thick] (bL1)--(bL2);
\draw[-stealth,black,very thick] (bL2)--(bL3);
\draw[-stealth,black,very thick] (bL3)--(tL3);
\draw[-stealth,black,very thick] (tL3)--(tL2);
\draw[-stealth,black,very thick] (tL2)--(tL1);
\draw[-stealth,gray,very thick, dashed] (tL1)--(tL0);
\draw[-stealth,gray,very thick, dashed] (tR3)--(tR2);
\draw[-stealth,black,very thick] (tR2)--(tR1);
\draw[-stealth,black,very thick] (tR1)--(tR0);
\draw[-stealth,black,very thick] (bR3) -- ({\rstart+3+0.4}, 0);
\draw[-stealth,black,very thick] (-0.4, 0) -- (bL0);
\draw[-stealth,black,very thick] ({\rstart+3+0.4}, 1) -- (tR3);
\draw[-stealth,black,very thick] (tL0) -- (-0.4, 1);
\end{tikzpicture}

    \vspace{1.5em}

\begin{tikzpicture}[scale=.9]
\pgfdeclarelayer{background}
\pgfsetlayers{background,main}

\def\gapsize{1.15}
\pgfmathsetmacro{\rstart}{3+\gapsize}

\node at (-0.6,.5) {\Large $\cdots$};
\node at ({\rstart+3+0.6},.5) {\Large $\cdots$};

\node[mynode,fill=dummycolor] (tL0) at (0,1) {};
\node[above,font=\small] at (0,1) {$d_{(i,3m-1)}$};
\node[mynode,fill=dummycolor] (tL1) at (1,1) {};
\node[mynode,fill=dummycolor] (tL2) at (2,1) {};
\node[mynode,fill=dummycolor] (tL3) at (3,1) {};
\node[above,font=\small] at (3,1) {$d_{(i,3m+2)}$};

\node[mynode,fill=dummycolor] (bL0) at (0,0) {};
\node[below,font=\small] at (0,0) {$d'_{(i,3m-1)}$};
\node[mynode,fill=dummycolor] (bL1) at (1,0) {};
\node[mynode,fill=dummycolor] (bL2) at (2,0) {};
\node[mynode,fill=dummycolor] (bL3) at (3,0) {};
\node[below,font=\small] at (3,0) {$d'_{(i,3m+2)}$};

\node[mynode,fill=dummycolor] (tR0) at ({\rstart},1) {};
\node[above,font=\small] at ({\rstart},1) {$d_{(i,0)}$};
\node[mynode,fill=dummycolor] (tR1) at ({\rstart+1},1) {};
\node[mynode,fill=dummycolor] (tR2) at ({\rstart+2},1) {};
\node[mynode,fill=dummycolor] (tR3) at ({\rstart+3},1) {};
\node[above,font=\small] at ({\rstart+3},1) {$d_{(i,3)}$};

\node[mynode,fill=dummycolor] (bR0) at ({\rstart},0) {};
\node[below,font=\small] at ({\rstart},0) {$d'_{(i,0)}$};
\node[mynode,fill=dummycolor] (bR1) at ({\rstart+1},0) {};
\node[mynode,fill=dummycolor] (bR2) at ({\rstart+2},0) {};
\node[mynode,fill=dummycolor] (bR3) at ({\rstart+3},0) {};
\node[below,font=\small] at ({\rstart+3},0) {$d'_{(i,3)}$};

\begin{pgfonlayer}{background}
\draw[-,lightgray] (tL0)--(tL1)--(tL2)--(tL3);
\draw[-,lightgray] (bL0)--(bL1)--(bL2)--(bL3);
\draw[-,lightgray] (tR0)--(tR1)--(tR2)--(tR3);
\draw[-,lightgray] (bR0)--(bR1)--(bR2)--(bR3);
\draw[-,lightgray] (tL0)--(bL0); \draw[-,lightgray] (tL1)--(bL1);
\draw[-,lightgray] (tL2)--(bL2); \draw[-,lightgray] (tL3)--(bL3);
\draw[-,lightgray] (tR0)--(bR0); \draw[-,lightgray] (tR1)--(bR1);
\draw[-,lightgray] (tR2)--(bR2); \draw[-,lightgray] (tR3)--(bR3);
\draw[-,lightgray] (tL0)--(bL1); \draw[-,lightgray] (tL1)--(bL2); \draw[-,lightgray] (tL2)--(bL3);
\draw[-,lightgray] (tR0)--(bR1); \draw[-,lightgray] (tR1)--(bR2); \draw[-,lightgray] (tR2)--(bR3);
\draw[-,lightgray] (bL0)--(tL1); \draw[-,lightgray] (bL1)--(tL2); \draw[-,lightgray] (bL2)--(tL3);
\draw[-,lightgray] (bR0)--(tR1); \draw[-,lightgray] (bR1)--(tR2); \draw[-,lightgray] (bR2)--(tR3);
\end{pgfonlayer}

\draw[-stealth,black,very thick] (bR0)--(tR0);
\draw[-stealth,black,very thick] (tR0)--(tR1);
\draw[-stealth,black,very thick] (tR1)--(bR1);
\draw[-stealth,black,very thick] (bR1)--(bR2);
\draw[-stealth,black,very thick] (bR2)--(tR2);
\draw[-stealth,gray,very thick, dashed] (tR2)--(bR3);  
\draw[-stealth,black,very thick] (bR3)--(tR3);
\draw[-stealth,black,very thick] (bL0)--(tL0);
\draw[-stealth,gray,very thick, dashed] (tL0)--(bL1);  
\draw[-stealth,black,very thick] (bL1)--(tL1);
\draw[-stealth,black,very thick] (tL1)--(tL2);
\draw[-stealth,black,very thick] (tL2)--(bL2);
\draw[-stealth,black,very thick] (bL2)--(bL3);
\draw[-stealth,black,very thick] (bL3)--(tL3);
\draw[-stealth,black,very thick] (tR3) -- ({\rstart+3+0.4}, 1);
\draw[-stealth,black,very thick] (-0.4, 0) -- (bL0);

\end{tikzpicture}
    \caption{Traversal options for dummy pair $i$. Top: Option~1, traversing $d'_i$ clockwise then $d_i$ counterclockwise (Option~2 is the mirror image). Bottom: Option~3, a clockwise zig-zag entering at $d'_{(i,0)}$ and exiting at $d_{(i,3m+2)}$ (Option~4 is the mirror image). Dashed arrows indicate where clause detours may be inserted. Thick arrows represent edges of $\pi$.}
\label{fig:options}
\end{figure}

Table~\ref{tab:cycle-options} lists how the entry and exit choices from the forward pass determine the  entry and exit in the return pass. 
\begin{table}[H]
\centering
\begin{tabular}{|l|l||l|l|}
\hline
$x_i$ exits & $x_{i+1}$ enters & Enter $d'_i$ & Exit $d_i$  \\
\hline
$3m+2$ & $3m+2$ & $0$    & $0$     \\

$0$    & $0$    & $3m+2$ & $3m+2$  \\
$0$    & $3m+2$ & $0$    & $3m+2$  \\
$3m+2$ & $0$    & $3m+2$ & $0$     \\
\hline
\end{tabular}
\caption{Option selection for dummy pair $i$, based on the forward-pass
transition $x_{(i,\text{exit})} \to x_{(i+1,\text{entry})}$.}
\label{tab:cycle-options}
\end{table}

In Options~1 and~2, dummy clause vertices $e_{(i,j)}$ are collected via detours during the traversal of $d_i$, and $e'_{(i,j)}$ via detours during the traversal of $d'_i$. In Option~3, the zig-zag is arranged so that at each clause position $j$, the path arrives at $d_{(i,3j-1)}$ on the unprimed ring. (In Option~4, the path analogously arrives at \ $d_{(i,3j)}$.) From there, the dashed diagonal hop is replaced by a clause detour executing the sequence $d_{(i,3j-1)} \to e_{(i,j)} \to e'_{(i,j)} \to d'_{(i,3j)} \to d_{(i,3j)},$ collecting both dummy clause vertices before resuming the zig-zag. When no clause detour
is needed at position $j$, the dashed diagonal hop $d_{(i,3j-1)} \to d'_{(i,3j)}$ is
taken directly. \end{proof}
\subsection{Non-Crossing Hamiltonian Cycle $\implies$ Satisfying Assignment}

We now prove the converse in the following theorem.

\begin{theorem}
   If $G_\phi^{\mathrm{cycle}}$ contains a non-crossing Hamiltonian cycle, then $\phi$ is satisfiable.
\end{theorem}

\begin{proof}
    Assume $\pi$ is a non-crossing Hamiltonian cycle in $G_\phi^{\mathrm{cycle}}$. Since the only edges incident to $s$ are the two forward edges $s \to v_{(1,\cdot)}$ and the two backward edges $d_{(1,\cdot)} \to s$, the cycle must use one edge of each type at $s$.  $\pi$ contains a subpath from $s$ to $t$, and the same argument as in Theorem~\ref{NCHP_backward} shows that any cross-ring clause detour would force a ring vertex to become an endpoint of $\pi$. So, the subpath from $s$ to $t$ is a non-crossing Hamiltonian path through the real rings, and therefore, by Theorem~\ref{NCHP_backward}, $\phi$ is satisfiable.
\end{proof}

Since $G_\phi^{\mathrm{cycle}}$ has $n(3m+3) + 2(n-1)(3m+3) + m(n+1) + 2m(n-1) + 2$ vertices, it can be constructed in polynomial time. NP-hardness follows from the reduction, and membership in NP is immediate since a Hamiltonian cycle can be verified in polynomial time. Thus, \textsc{NCHC} is NP-complete.

\section{Conclusion}

We provided reductions from 3-SAT to both NCHP and NCHC, establishing NP-completeness without relying on planarity of the underlying graph. The graph embeddings produced by our reductions are sparse, so a natural direction for future work is to determine the maximum number of edges that can be added while preserving NP-hardness. Progress on this question would shed light on the complexity of related problems, such as non-crossing Hamiltonian paths and cycles in embedded complete bipartite graphs.

\nocite{*}
\bibliographystyle{plain}
\bibliography{sample}

\end{document}